\documentclass[num-refs]{wiley-article}

\usepackage{amsfonts,amssymb,amscd}
\usepackage{latexsym}
\usepackage{tcolorbox}
\usepackage[all]{xy}
\usepackage{tikz}
\usepackage{placeins}
\usepackage{graphicx}
\usepackage{float}
\usepackage{subfig}

\usepackage{algorithm}
\usepackage{algorithmicx}
\usepackage{algcompatible}
\usepackage[noend]{algpseudocode}
\usepackage{algpascal}

\newtheorem{claim}[theorem]{Claim}

\papertype{Research Article}

\newcommand {\npc}     {\textsc{NP}\text{-complete} }
\newcommand {\nph}     {\textsc{NP}\text{-hard} }

\title{Minimum Reload Cost Cycle Cover in Complete Graphs}

\author[1]{Yasemin B\"{u}y\"{u}k\c{c}olak}
\author[2]{Didem G\"{o}z\"{u}pek}
\author[1]{Sibel \"{O}zkan}

\affil[1]{Department of Mathematics, Gebze Technical University, Kocaeli, Turkey}
\affil[2]{Department of Computer Engineering, Gebze Technical University, Kocaeli, Turkey}

\corraddress{Yasemin B\"{u}y\"{u}k\c{c}olak, Department of Mathematics, Gebze Technical University, Kocaeli, Turkey}
\corremail{y.buyukcolak@gtu.edu.tr}

\fundinginfo{TUBITAK-CNRS, Grant Number: 114E731}

\runningauthor{B\"{u}y\"{u}k\c{c}olak, G\"{o}z\"{u}pek and \"{O}zkan}

\begin{document}

\maketitle

\begin{abstract}
The reload cost refers to the cost that occurs along a path on an edge-colored graph when it traverses an internal vertex between two edges of different colors.  Galbiati et al. \cite{GalbiatiGualandiMaffioli} introduced the \emph{Minimum Reload Cost Cycle Cover} problem, which is to find a set of vertex-disjoint cycles spanning all vertices with minimum reload cost. They proved that this problem is strongly \nph and not approximable within $1/\epsilon$ for any $\epsilon>0$ even when the number of colors is $2$, the reload costs are symmetric and satisfy the triangle inequality. In this paper, we study this problem in complete graphs having equitable or nearly equitable $2$-edge-colorings. By showing the existence of a monochromatic cycle cover we prove that the minimum reload cost is  zero on complete graphs $K_n$ with an equitable $2$-edge-coloring except possibly $n = 4$ or with a nearly equitable $2$-edge-coloring except possibly for $n \leq 13$. Furthermore, we provide a polynomial-time algorithm that constructs a monochromatic cycle cover in complete graphs $K_n$ with an equitable $2$-edge-coloring except possibly for $n=4$. This algorithm also finds a monochromatic cycle cover in complete graphs with a nearly equitable $2$-edge-coloring except for some special cases.

\end{abstract}

\section{Introduction}

Edge-colored graphs can be used to model various network design problems. In this work, we consider an optimization problem with the \emph{reload cost} model. The reload cost occurs along a path on an edge-colored graph while traversing an internal vertex via two consecutive edges of different colors. That is, the reload cost depends only on the colors of the incident traversed edges. In addition, the reload cost of a path or a cycle on an edge-colored graph is the sum of the reload costs that arise from traversing its internal vertices between edges of different colors. Because of practical reasons, it is generally assumed that the reload costs are symmetric and satisfy the triangle inequality. The reload cost concept is used in many areas such as transportation networks, telecommunication networks, and energy distribution networks. For instance, in a cargo transportation network, each carrier can be represented by a color and the reload costs arise only at points where the carrier changes, i.e., during transition from one color to another. In telecommunication networks, the reload costs arise in several settings. For instance, switching among different technologies such as cables, fibers, and satellite links or switching between different providers such as different commercial satellite providers in satellite networks correspond to reload costs. In energy distribution networks, the reload cost corresponds to the loss of energy while transferring energy from one form to another one, such as the conversion of natural gas from liquid to gas form.

Although the reload cost concept has significant applications in many areas, only few papers about this concept have appeared in the literature. Wirth and Steffan \cite{WirthSteffan}, and Galbiati \cite{Galbiati} studied the minimum reload cost diameter problem, which is to find a spanning tree with minimum diameter with respect to reload cost. Amaldi et al. \cite{AmaldiGalbiatiMaffioli} presented several path, tour, and flow problems under the reload cost model. They also focused on the problem of finding a spanning tree that minimizes the total reload cost from a source vertex to all other vertices.
The works in \cite{GGM11, GOP+16, FPT-by-tw-Delta, GozupekSVZ14} focused on the minimum changeover cost arborescence problem, which is to find a spanning arborescence rooted at a given vertex such that the total reload cost is minimized. The work in  \cite{GoSh16}, on the other hand, focused on problems related to finding a proper edge coloring of the graph so that the total reload cost is minimized.

Galbiati et al. \cite{GalbiatiGualandiMaffioli} introduced the \emph{Minimum Reload Cost Cycle Cover} (\textsc{MinRC3}) problem, which is to find a set of vertex-disjoint cycles spanning all vertices with minimum total reload cost.  They proved that it is strongly \nph and not approximable within $1/\epsilon$ for any $\epsilon > 0$ even when the number of colors is $2$, the reload costs are symmetric and satisfy the triangle inequality. In this work we focus on a special case of the \textsc{MinRC3} problem, namely \textsc{MinRC3} in complete graphs. Our primary motivation is to avoid the feasibility issue, since a complete graph of any order has a cycle cover. We first show that the \textsc{MinRC3} problem is strongly \nph and is not approximable within $1/\epsilon$ for any $\epsilon > 0$ for complete graphs, even when the reload costs are symmetric. We are then interested in the \textsc{MinRC3} problem in complete graphs having an \emph{equitable $2$-edge-coloring}, which is an edge-coloring with two colors such that for each vertex $v \in V(G)$, $||c_1(v)| -|c_2(v)|| \leq 1$, where $c_i(v)$ is the set of edges with color $i$ that are incident to $v$. To the best of our knowledge, this paper is the first one focusing on the \textsc{MinRC3} problem in a special graph class. In particular, we present the first positive (polynomial-time solvability) result for this problem.

Feasibility of equitable edge-colorings received some attention in the literature. In 2008, Xie et al. \cite{Xie} showed that the problem of finding whether an equitable $k$-edge-coloring exists is \npc in general. Indeed, if $k = \Delta$, where $\Delta$ is the maximum degree of the given graph, then this problem becomes equivalent to the well-known \npc problem of classifying Class-$1$ graphs. In 1994, Hilton and de Werra \cite{HiltonWerra} proved the following sufficiency condition on equitable $k$-edge-colorings: if $k \geq 2$ and $G$ is a simple graph such that no vertex in $G$ has degree equal to a multiple of $k$, then $G$ has an equitable $k$-edge-coloring. In 1971, de Werra \cite{Werra} found the following necessary and sufficient condition to have an equitable $2$-edge-coloring in a connected graph: a connected graph $G$ has an equitable $2$-edge-coloring if and only if it is not a connected graph with an odd number of edges and all vertices having an even degree. Furthermore, a \emph{nearly equitable $k$-edge-coloring} is an edge-coloring with $k$ colors such that for each vertex $v \in V(G)$ and for each pair of colors $i, j \in \{1,2,...,k\}$, $||c_i(v)| -|c_j(v)|| \leq 2$, where $c_i(v)$ is the set of edges with color $i$ that are incident to $v$. The notion of nearly equitable edge-coloring was introduced in 1982 by Hilton and de Werra \cite{HiltonWerra82}, who also proved that for each $k\geq 2$ any graph has a nearly equitable $k$-edge-coloring.

In this paper, we focus on  the \textsc{MinRC3} problem in complete graphs having an equitable or nearly equitable $2$-edge-coloring. Recall that the reload cost is zero between two edges of the same color. The reload cost of a monochromatic cycle, i.e., a cycle having all edges with the same color, is clearly zero. We then investigate the existence of a monochromatic cycle cover in complete graphs having an equitable or nearly equitable $2$-edge-coloring. In the literature, there exist various results about covering a $k$-edge-colored graph with monochromatic subgraphs such as cycles, paths, and trees (see \cite{ErdosGyarfasPyber, Gyarfas}).  In 1983, Gy\'{a}rf\'{a}s \cite{Gyarfas} proved that the vertex set of any $2$-edge-colored complete graph can be covered by two monochromatic cycles that have different colors and intersect in at most one vertex. In 2010, Bessy and Thomass\'{e} \cite{BessyThomasse} proved that the vertex set of any $2$-edge-colored complete graph can be partitioned into two monochromatic cycles having different colors, i.e., it has a vertex-disjoint monochromatic cycle cover with two different colors. However, unlike in the \textsc{MinRC3} problem, a vertex ($K_1$) and an edge ($K_2$) are considered to be cycles in almost all works in the literature (including \cite{BessyThomasse}) about monochromatic cycle covers. Clearly, both a vertex and an edge have zero reload cost, yet in this paper, we do not allow cycles to have less than three vertices.

We prove in this paper that except possibly for $n \leq 13$ in a complete graph $K_n$ with a nearly equitable $2$-edge-coloring, there exists a cycle cover that is either a monochromatic Hamiltonian cycle or consists of exactly two monochromatic cycles on the same color with sizes differing by at most one; therefore, the value of the minimum reload cost cycle cover is zero in such a case. In addition, we show that except possibly for $n = 4$ there exists a monochromatic cycle cover in complete graphs $K_n$ with an equitable edge coloring. Our constructive proof leads to a polynomial-time algorithm to solve the \textsc{MinRC3} problem on complete graphs with an equitable $2$-edge coloring. Our proof also leads to a polynomial-time algorithm to solve this problem on complete graphs with a nearly equitable $2$-edge coloring except for some special cases.

\section{Preliminaries}\label{sec:prelim}

An undirected graph  $G = (V(G),E(G))$ is given by a pair of a vertex set $V(G)$ and an edge set $E(G)$, which consists of $2$-element subsets $\{u,v\}$ of $V(G)$. An edge $\{u,v\}$ between two vertices $u$ and $v$ will be denoted by $uv$ in short. In this work, we consider only simple graphs, i.e., graphs without loops or multiple edges. The \emph{order} of $G$  is denoted by $|V(G)|$ and the \emph{degree} of a vertex $v$ of $G$ is denoted by $d(v)$. In addition, $\delta(G)$ and  $\Delta(G)$ denote the \emph{minimum} and \emph{maximum degree} of $G$, respectively. When the graph $G$ is clear from the context, we omit it from the notations and write $V$, $E$, $\delta$ and $\Delta$.

Given two graphs $G=(V,E)$ and $G'=(V',E')$, if $G$ is isomorphic to $G'$, we denote it by $G \cong G'$. We define the \emph{union} $G\cup G'$ of $G$ and $G'$ as the graph obtained by the union of their vertex and edge sets, i.e., $G \cup G' = (V \cup V', E \cup E')$. When $V$ and $V'$ are disjoint, their union is referred to as the \emph{disjoint union} and denoted by $G + G'$.  The \emph{join} $G \vee G'$ of $G$ and $G'$ is the disjoint union of graphs $G$ and $G'$ together with all the edges joining $V$ and $V'$. Formally,  $G \vee G' = (V \cup V', E \cup E' \cup \{V \times V'\})$. The \emph{complement} of a graph $G =(V,E)$ is the graph $\overline{G}=(V, \overline{E})$ (the same vertex set $V$ but whose edge set $\overline{E}$ consists of $2$-element subsets of $V$ that are not in $E$). That is, $E \cap \overline{E} = \emptyset$ and $E \cup \overline{E}$ contains all possible edges on the vertex set of $G$.

The \emph{closure} of a graph $G$ with $n$ vertices, denoted by $cl(G)$, is the graph obtained from $G$ by repeatedly adding edges between nonadjacent vertices whose degrees sum to at least $n$, until no such vertices exist. The \emph{degree sequence} of a graph $G$ is the nondecreasing sequence of its vertex degrees. A graph is \emph{r-regular} if all of its vertices have degree $r$. We say that a graph $H$ is an \emph{r-factor} of a graph $G$ when $V(H)=V(G)$ and $H$ is $r$-regular. Notice here that a cycle cover of a graph $G$ is equivalent to a $2$-factor of $G$.

An \emph{independent set} in a graph $G$ is a subset of pairwise nonadjacent vertices in $V(G)$. A \emph{maximum independent set} is an independent set of largest size for a given graph $G$. The size of a maximum independent set is called the \emph{independence number} of $G$ and is denoted by $\alpha(G)$.
Besides, a \emph{clique} of $G$ is a subset of vertices of $G$ whose induced subgraph is a complete graph. The following lemma is used in our arguments:
\begin{claim}\label{independencenumber}
For a graph $G$ on $n$ vertices with $\delta \leq n/2$ and  $\Delta \leq n - \delta$, the independence number $\alpha$ of $G$ satisfies the inequality $\alpha \leq n - \delta$ and the equality holds only for the complete bipartite graph $K_{n-\delta, \delta}$ with $\Delta = n-\delta$.
\end{claim}
\begin{proof}[Proof of Claim \ref{independencenumber}]
Assume to the contrary that $I$ is  an independent set of $G$ with size greater than $n - \delta$ and let $J$ be the remaining vertices in $G$, i.e., $|I| > n - \delta$ and $|J| < \delta$.  Each vertex in $I$ can be adjacent only to the vertices in $J$ since $I$ is an independent set in $G$. However, then each vertex in $I$ has degree less than $\delta$ in $G$, which is a contradiction  since $\delta$ is the minimum degree. When $|I| = n - \delta$ and $|J| = \delta$, each vertex in $I$ must be adjacent to every vertex in $J$ to attain the minimum degree $\delta$; moreover, each vertex in $J$ can be adjacent only to the vertices in $I$ since $\Delta \leq n - \delta$. Hence, $\alpha=n-\delta$ only when $G$ is $K_{n-\delta, \delta}$ with $\Delta = n-\delta$ as desired.
\end{proof}

A \emph{cycle} on $n$ vertices is denoted by $C_n$. A \emph{cycle cover} of a graph $G$ is a collection of cycles such that every vertex in $G$ is contained in at least one such cycle. If the cycles of the cover have no vertices in common, the cover is called \emph{vertex-disjoint}. Unless otherwise stated, cycle covers are always assumed to be vertex-disjoint in this work.
A \emph{Hamiltonian cycle }of a graph $G$ is a cycle passing through every vertex of $G$ exactly once, and a graph $G$ containing a Hamiltonian cycle is called \emph{Hamiltonian}. Some fundamental results on hamiltonicity used in this paper are as follows:

\begin{theorem}\label{di}\textbf{(Dirac \cite{Dirac})}
If $G$ is a graph of order $n \geq 3$ such that $\delta(G) \geq  n/2$, then
$G$ is Hamiltonian.
\end{theorem}

\begin{theorem}\label{buyuk}\textbf{(Büyükçolak et al. \cite{Buyukcolak})}
Let $G$ be a connected graph of order $n \geq 3$ such that $\delta(G) \geq \lfloor n/2\rfloor$. Then $G$ is Hamiltonian
unless $G$ is a graph $K_{\lceil n/2 \rceil} \cup K_{\lceil n/2 \rceil}$ with one common vertex or a graph $\overline{K}_{\lceil n/2 \rceil} \vee G_{\lfloor n/2\rfloor}$  for odd $n$, where $G_n$ is a not necessarily connected simple graph on $n$ vertices.
\end{theorem}

\begin{theorem}\label{bo}\textbf{(Bondy-Chv\'{a}tal \cite{Bondy})}
A graph $G$ is Hamiltonian if and only if its closure $cl(G)$ is Hamiltonian.
\end{theorem}

\begin{theorem}\label{ch}\textbf{(Chv\'{a}tal \cite{Chvatal})}
Let $G$ be a simple graph with degree sequence $(d_1, d_2, ... , d_n)$, where $d_1 \leq d_2 \leq \cdots \leq d_n$. If there is no $m < n/2$ such that $d_m \leq m$ and $d_{n-m} < n - m$, then $cl(G)$ is a complete graph and therefore $G$ is Hamiltonian.
\end{theorem}

\begin{theorem}\label{n-w2}\textbf{(Nash-Williams \cite{NashWilliams})}
Let $G$ be  $2$-connected graph of order $n$ with  $\delta \geq \max\{(n+2)/3,\alpha(G)\}$. Then $G$ is Hamiltonian.
\end{theorem}

\begin{theorem}\label{m-m}\textbf{(Moon-Moser \cite{MoonMoser})}
Let $G$ be a bipartite graph with two disjoint vertex sets $V_1$ and $V_2$ such that $|V_1|=|V_2|=m$. If $\min \{d(u) + d(v)~|~ u\in V_1, v\in V_2, u \emph{\textsc{ and }} v \emph{\textsc{ are nonadjacent}}\} \geq m+1$, then $G$ is Hamiltonian.
\end{theorem}

A \emph{k-edge-coloring} of a graph $G$ is an assignment of $k$ colors to edges of $G$, which is represented by a mapping $\chi: E(G) \rightarrow \mathcal{C}$, where $\mathcal{C} = \{c_1, c_2, ... , c_k\}$ is a set of $k$ colors. Given a $k$-edge-coloring of $G$ with $k$ colors $c_1, . . . , c_k$, $c_i(v)$ denotes the set of edges incident to $v$ colored with $c_i$ for $v \in V(G)$, where $1\leq i\leq k$.
A \emph{reload cost function} is a function $\rho : \mathcal{C} \times \mathcal{C} \rightarrow \mathbb{N}_0$ such that for all pairs of colors $c_1, c_2 \in \mathcal{C}$,
\begin{enumerate}
  \item $c_1 = c_2 \Rightarrow \rho(c_1, c_2) = 0$,
  \item $c_1 \neq c_2 \Rightarrow \rho(c_1, c_2) > 0$.
\end{enumerate}
The reload cost incurs at a vertex while traversing two consecutive edges of different colors and $\rho(e_1, e_2) = \rho(\chi(e_1),\chi(e_2))$, where $e_1$ and $e_2$ are incident edges. {The reload cost is said to be \emph{symmetric} if $\rho(e_1, e_2) = \rho(e_2, e_1)$ and to satisfy the \emph{triangle inequality} if $\rho(e_1, e_3) \leq \rho(e_1, e_2) + \rho(e_2, e_3)$ for mutually incident edges $e_1$, $e_2$ and $e_3$.}
The reload cost of a path is the sum of the reload costs that occur at its internal vertices, i.e., $\rho(P) = \rho(e_{1},e_2) + \rho(e_{2},e_3) + \cdots + \rho(e_{n-1},e_n)$, where $P = (e_1 - e_2 - \cdots - e_n)$ is a path of length $n-1$. The reload cost of a cycle is $\rho(C)=\rho(e_{1},e_2) + \rho(e_{2},e_3) + \cdots + \rho(e_{n-1},e_n) + \rho(e_n,e_1)$, where $C$ is a cycle consisting of edges $e_1, e_2,\dots, e_n$ in this cyclic order. Note that a monochromatic path or cycle, i.e., a path or cycle having all edges with the same color, clearly has zero reload cost. Besides, the reload cost of a cycle cover is the sum of the reload costs of each cycle component of the cycle cover, i.e., $\rho(\textbf{C}) = \rho(C_1) + \rho(C_2) + \rho(C_3) + \cdots + \rho(C_n)$, where $\textbf{C}= C_1 + C_2 + \cdots + C_n$.

The \emph{Minimum Reload Cost Cycle Cover} (\textsc{MinRC3}) problem is an optimization problem which aims to span all vertices of an edge-colored graph by a set of vertex-disjoint cycles with minimum reload cost. Formally,

\begin{tcolorbox}
\textsc{MinRC3} $(G, \mathcal{C}, \chi, \rho)$ \\
\textbf{Input:} A graph $G = (V, E)$ with an edge coloring function $\chi : E \rightarrow \mathcal{C}$ and a reload cost function
$\rho : \mathcal{C} \times \mathcal{C} \rightarrow \mathbb{N}_0$. \\
\textbf{Output:} A cycle cover $\textbf{C}$ of $G$. \\
\textbf{Objective:} Minimize $\rho(\textbf{C})$.
\end{tcolorbox}
\noindent
The previous results on  \textsc{MinRC3} are as follows:
\begin{theorem}\label{gua1}\textbf{(Galbiati et al. \cite{GalbiatiGualandiMaffioli})}
\textsc{MinRC3} is strongly \nph even if the number of colors is $2$, the reload costs are symmetric, and satisfy the triangle inequality.
\end{theorem}
\begin{corollary}\label{gua2}\textbf{(Galbiati et al. \cite{GalbiatiGualandiMaffioli})}
\textsc{MinRC3} is not approximable within $1/\epsilon$ , for any $\epsilon > 0$, even if the number of colors is $2$, the reload costs are symmetric, and satisfy the triangle inequality.
\end{corollary}

A monochromatic cycle cover is composed of cycles such that the colors of the edges of a particular cycle are the same; however, the colors of edges in different cycles may differ in general. In this work, we investigate the \textsc{MinRC3} problem in equitably or nearly equitably $2$-edge-colored complete graphs and prove that the minimum reload cost is zero in such graphs except for some special cases by constructing a monochromatic cycle cover in a single color.

\section{\textsc{MinRC3} in Complete Graphs}

By Theorem \ref{gua1} and Corollary \ref{gua2}, we already know that the \textsc{MinRC3} problem is \nph in the strong sense and not approximable within $1/\epsilon$ for any $\epsilon > 0$, even when the number of colors is $2$. In the following theorem, we prove a hardness result for complete graphs:

\begin{theorem}\label{cor:inapproxMinRC3}
The \textsc{MinRC3} problem is strongly \nph  and not approximable within $1/\epsilon$ for any $\epsilon > 0$ for complete graphs even if the reload costs are symmetric.
\end{theorem}
\begin{proof}
The proof is by reduction from the problem itself. Given an instance $I$ of $\textsc{MINRC3}$ $(G, \mathcal{C}, \chi, \rho)$, where $G=(V(G), E(G))$, we construct an instance $I'$ of $\textsc{MINRC3}$ $(G', \mathcal{C}', \chi', \rho')$ as follows: $G'=(V(G), E(G'))$ is a complete graph such that $E(G')=E(G) \cup (\cup_{uv \notin E(G)} uv)$ and $\mathcal{C'}=\mathcal{C} \cup (\cup_{uv \notin E(G)}\chi(uv))$. For all $uv \notin E(G)$ and $c \in \mathcal{C'} \setminus \chi(uv)$, we set $\rho(\chi(uv), c) = \rho(c,\chi(uv))=M$, where $M$ is a very large integer. In other words, for every $uv\notin E(G)$, $\chi(uv)$ is a new color in $G'$ having a very large reload cost value with all other colors in $\mathcal{C}'$. This reduction shows that $I$ is a satisfiable instance of $\textsc{MINRC3}$ if and only if $I'$ is a satisfiable instance of $\textsc{MINRC3}$ in complete graphs. Furthermore, in the case where $I$ is a satisfiable instance, we clearly have $OPT(G)=OPT(G')$, where $OPT(G)$ and $OPT(G')$ denote the reload cost of an optimum solution of $G$ and $G'$, respectively. Let $A'$ be a $1/\epsilon$ approximation algorithm for the \textsc{MinRC3} problem in complete graphs for some $\epsilon >0$. Then since $A'(I') \le (OPT(I')/\epsilon)=OPT(I)/\epsilon$, $I'$ is also a $1/\epsilon$ approximation algorithm for the \textsc{MinRC3} problem in general, contradicting Corollary \ref{gua2}. Hence, the theorem holds.
\end{proof}

Having proved that \textsc{MinRC3} is in general inapproximable within $1/\epsilon$ for any $\epsilon > 0$ in complete graphs, now  we investigate \textsc{MinRC3} in complete graphs with equitable and nearly equitable $2$-edge-colorings.

\subsection{Complete Graphs with Equitable $2$-Edge-Coloring}\label{subsec:equitable}
    Since a monochromatic cycle cover has zero reload cost, it is sufficient to show that there exists a partition of vertices of a complete graph having an equitable $2$-edge-coloring into monochromatic vertex-disjoint cycles.

The following lemma given in \cite{Werra} implies the existence of an equitable $2$-edge-coloring in complete graphs except $K_{4k+3}$, where $k \ge 0$:

\begin{lemma}\label{lem-werra}
\textbf{\cite{Werra}} A connected graph $G$ has an equitable $2$-edge-coloring if and only if it is not a connected graph with odd number of edges and all vertices having an even degree.
\end{lemma}

\begin{corollary}\label{lem-odd3no}
A complete graph has an equitable $2$-edge-coloring if and only if it is not a complete graph $K_{4k+3}$ with $k \geq 0$.
\end{corollary}

Now, we analyze the cases of $K_n$ for even $n$ and odd $n$ separately. For odd $n$, by Corollary \ref{lem-odd3no} it suffices to examine complete graphs with order $n=4k+1$ . The following lemma shows that a complete graph $K_{4k+1}$ having an equitable $2$-edge-coloring has a monochromatic Hamiltonian cycle in both colors:

    \begin{lemma}\label{lem-odd1}
    For a complete graph $K_{4k+1}$, where $k \geq 1$, with an equitable $2$-edge-coloring, there exists a monochromatic cycle cover of the form $C_{4k+1}$ for both colors; in other words, there exist monochromatic Hamiltonian cycles in both colors.
    \end{lemma}

    \begin{proof}
    Let $\chi$ be an equitable $2$-edge-coloring in the complete graph $K_{4k+1}$, $k \geq 1$, with colors, say red and blue. In an equitable $2$-edge-coloring of $K_{4k+1}$, each vertex is incident to $2k$ red edges and $2k$ blue edges. Consider the color induced subgraphs $K_{4k+1}^r$ and $K_{4k+1}^b$ in $K_{4k+1}$ for red and blue colors, respectively. Both $K_{4k+1}^r$ and $K_{4k+1}^b$ are $2k$-regular graphs on $4k+1$ vertices. Note that a $2k$-regular graph on $4k+1$ vertices cannot be disconnected because otherwise each component has to have at least $2k+1$ vertices, contradicting the specified order $4k+1$. Both $K_{4k+1}^r$ and $K_{4k+1}^b$ are connected $2k$-regular graphs on $4k+1$ vertices. Hence, they have Hamiltonian cycles by Theorem \ref{buyuk} since they are neither $K_{2k+1} \cup K_{2k+1}$ with one common vertex nor $\bar{K}_{2k+1} \vee G_{k}$. Therefore, there exists a monochromatic cycle cover of the form $C_{4k+1}$ in both colors, as desired.
    \end{proof}

 For even $n \ge 6$, the following lemma shows that a complete graph $K_{n}$ having an equitable $2$-edge-coloring has a monochromatic cycle cover with at most two cycles having the same size and the same color:

    \begin{lemma}\label{lem-even}
    For a complete graph $K_{2k}$, $k \geq 3$, with an equitable $2$-edge-coloring, there exists a monochromatic cycle cover in a single color of the form $C_k + C_k$ or $C_{2k}$. In particular, there exists a cycle cover $C_k + C_k$  in a single color if $K_{2k}$ has a disconnected or $1$-connected color induced subgraph, and a Hamiltonian cycle $C_{2k}$ otherwise.
    \end{lemma}

    \begin{proof}
    Let $\chi$ be an equitable $2$-edge-coloring of the complete graph $K_{2k}$  with red and blue colors. In an equitable $2$-edge-coloring of $K_{2k}$, each vertex is incident to either $k$ red edges and $k-1$ blue edges or $k$ blue edges and $k-1$ red edges. We consider the subgraphs $K_{2k}^r$ and $K_{2k}^b$ in $K_{2k}$ induced by red and blue colors, respectively. For both of them, the minimum degree is at least $k-1$ and the maximum degree is at most $k$; i.e., $k-1 \leq \delta \leq \Delta \leq k$.

    Note that the only disconnected graph on $2k$ vertices with $\delta \geq k-1$  is the disjoint union of two $K_k$, i.e., $K_k + K_k$, which is a $(k-1)$-regular graph. Let $K_{2k}^r$ be a disconnected graph on $2k$ vertices with $\delta \geq k-1$; i.e., $K_{2k}^r \cong K_k + K_k$. Since both components of $K_{2k}^r$ are complete graphs, $K_{2k}^r$ has a cycle cover of the form $C_k + C_k$. On the other hand, $K_{2k}^b$ is a complete bipartite graph $K_{k,k}$ since it is the complement of $K_k + K_k$. Clearly, $K_{k,k}$ has a Hamiltonian cycle by Theorem \ref{di}.  In this case, we therefore have a monochromatic cycle cover in both induced subgraphs $K_{2k}^r$ and $K_{2k}^b$.

    We now suppose that both $K_{2k}^r$ and $K_{2k}^b$ are connected graphs with $\delta \geq k-1$ and $\Delta \leq k$. Assume that  both $K_{2k}^r$ and $K_{2k}^b$ are regular graphs, i.e., $k-1 \leq \delta = \Delta \leq k$ for both graphs. Hence, one of the graphs is a $k$-regular graph on $2k$ vertices, whereas the other is a $(k-1)$-regular graph on $2k$ vertices. By Theorem \ref{di}, there is a Hamiltonian cycle on the $k$-regular graph on $2k$ vertices. Therefore, in the case where color induced subgraphs are regular, we have a monochromatic cycle cover of the form $C_{2k}$ in the $k$-regular subgraph induced by one of the colors.

    We then consider the case where both $K_{2k}^r$ and $K_{2k}^b$ are connected and are not regular graphs. In other words, both of them have $\delta = k-1$ and $\Delta = k$. Let the degree sequences of $K_{2k}^r$ and $K_{2k}^b$ be $(r_1, r_2, ..., r_{2k})$ and  $(b_1, b_2, ..., b_{2k})$, respectively, where $r_1 \leq r_2 \leq \cdots \leq r_{2k}$ and $b_1 \leq b_2 \leq \cdots \leq b_{2k}$. Since $\delta = k-1$ and $\Delta = k$ for both graphs, we have $k-1 \leq r_i, b_i \leq k$ for $1 \leq i \leq 2k$. Notice that at least half of the vertices in one of the induced subgraphs $K_{2k}^r$ and $K_{2k}^b$ have degree $k$, since these subgraphs are complements of each other. Without loss of generality (w.l.o.g.), assume that at least half of the vertices of $K_{2k}^r$ have degree $k$, i.e., $r_i = k$ for $i \geq k+1$. Let us consider the degrees of the remaining vertices, i.e., vertices of degree $k-1$, in $K_{2k}^r$:
    \begin{enumerate}
    \item Assume that the number of vertices having degree $k-1$ is less than $k-1$, i.e., $r_{k-1} = k$. Then the closure of $K_{2k}^r$ is a complete graph; therefore, by Theorem \ref{bo} $K_{2k}^r$ has a Hamiltonian cycle.
    \item  Assume that the number of vertices having degree $k-1$ is at least $k-1$, i.e., $r_{k-1} = k-1$. If there exists a pair of nonadjacent vertices $u$ and $v$ both having degree $k$, then the closure of $K_{2k}^r$ must contain the edge $uv$ by definition. The degrees of $u$ and $v$ become $k+1$ and then they must be adjacent to all other vertices in the closure of $K_{2k}^r$. Iteratively adding edges between nonadjacent vertices whose degrees sum to at least $n=2k$, we obtain the complete graph $K_{2k}$ as the closure of $K_{2k}^r$. Then, $K_{2k}^r$ has a Hamiltonian cycle by Theorem \ref{bo}. Otherwise, i.e., there is no pair of nonadjacent vertices both having degree $k$, then all vertices having degree $k$ are adjacent to each other. It follows that the vertices of degree $k$ form a clique of size $k$ or $k+1$ in $K_{2k}^r$ depending on the value of $r_k$. If the vertices of degree $k$ form a clique of size $k+1$ in $K_{2k}^r$, i.e., $r_k = k$, it contradicts the fact that $K_{2k}^r$ is a connected graph. Then, there are $k$ vertices having degree $k-1$ in $K_{2k}^r$, i.e., $r_k = k-1$. Hence, we deduce that there are $k$ vertices having degree $k$ in $K_{2k}^r$ and these vertices form a clique of size $k$ in $K_{2k}^r$. It implies that all vertices having degree $k-1$ in $K_{2k}^b$ form an independent set of size $k$ in $K_{2k}^b$. Besides, $k$ must be even in order to satisfy the relation {$2E(K_{2k}^r) = \sum_{i=1}^{2k} d(v_i) = k(2k-1)$} for $v_i \in V(K_{2k}^r)$. That is, $k \geq 4$. By Lemma \ref{independencenumber}, we have $\alpha \leq n - \delta = 2k-(k-1) = k+1$ in $K_{2k}^b$. We then observe that the set of $k$ vertices of degree $k-1$ in $K_{2k}^b$ is indeed a maximum independent set of size $k$ in $K_{2k}^b$, i.e., $\alpha = k$ in $K_{2k}^b$. Otherwise, i.e., $\alpha =k+1$ then $\Delta = k+1$, which is a contradiction.

        \indent
        Let us consider the spanning bipartite subgraph $\widetilde{K_{2k}^b}$ of $K_{2k}^b$ with the partite sets $V_1$ and $V_2$, which consist of vertices of degree $k-1$ and $k$ in $K_{2k}^b$, respectively. Notice that $|V_1| = |V_2| = k$ and there is no edge among the vertices of $V_1$ in $K_{2k}^b$. Hence,$\widetilde{K_{2k}^b}$ is obtained by removing all edges among the vertices of $V_2$ in $K_{2k}^b$, i.e., all edges among the vertices having degree $k$.
        Furthermore, $\widetilde{K_{2k}^b}$ contains $k(k-1) = k^2 - k$ edges since the vertices of $V_1$ have degree $k-1$ in $\widetilde{K_{2k}^b}$, whereas $K_{2k}^b$ contains $k^2-k/2$ edges since the vertices of $V_2$ have degree $k$ in $K_{2k}^b$. It means that $\widetilde{K_{2k}^b}$ is obtained by removing exactly $k/2$ edges, which join the vertices of degree $k$, from ${K_{2k}^b}$ where $k$ is even. Hence, one can observe that there is no isolated vertex of $V_2$ in $\widetilde{K_{2k}^b}$ since the vertices of $V_1$ form a maximum independent set in $K_{2k}^b$. Notice that the vertices in partite set $V_2$ have minimum degree at least $k/2$ in $\widetilde{K_{2k}^b}$ since we remove at most $k/2$ edges from a vertex of degree $k$ in $K_{2k}^b$. Hence, there is no leaf vertex of $V_2$ in $\widetilde{K_{2k}^b}$ since $k \geq 4$. It follows that for any nonadjacent vertices $u \in V_1$ and $v \in V_2$, we have $\min \{d(u) + d(v)\} = (k - 1) + k/2 \geq k+1$ in $\widetilde{K_{2k}^b}$ where $k \geq 4$. Since  $|V_1| = |V_2| = k$, $\widetilde{K_{2k}^b}$ has a Hamiltonian cycle $C_{2k}$ by Theorem \ref{m-m} , which is also a Hamiltonian cycle in $K_{2k}^b$. In this case, we therefore have a monochromatic cycle cover of the form $C_{2k}$ in $K_{2k}^b$.
    \end{enumerate}
    \end{proof}

We now combine Corollary \ref{lem-odd3no}, Lemmata \ref{lem-odd1} and \ref{lem-even} in the following way:

\begin{theorem}\label{combo}
 For $n \geq 5$, a complete graph $K_{n}$ with an equitable $2$-edge-coloring has a monochromatic cycle cover in a single color with at most two cycles. In particular, such a graph contains two cycles of the same size and the same color, or a monochromatic Hamiltonian cycle.
\end{theorem}

By Theorem \ref{combo}, we obtain the first main result of this section as follows:

\begin{corollary}\label{main}
For $n \geq 5$, the solution of the \textsc{MinRC3} problem equals zero for $K_n$ having an equitable $2$-edge-coloring.
\end{corollary}

\begin{remark}
In $K_4$ with an equitable $2$-edge-coloring, the only case where the solution of the \textsc{MinRC3} problem is nonzero is when both colors induce a path on three edges.
\end{remark}

\subsection{Complete Graphs with a Nearly Equitable $2$-Edge-Coloring}\label{subsec:nearlyequitable}
    
We now analyze the \textsc{MinRC3} problem in complete graphs having a nearly equitable $2$-edge-coloring. Note that every equitable $2$-edge-coloring is indeed a nearly equitable $2$-edge-coloring. Then, we only need to study the \textsc{MinRC3} problem in complete graphs having a nearly equitable $2$-edge-coloring that is not an equitable $2$-edge-coloring, say a \emph{sharp nearly equitable $2$-edge-coloring}.

In the following lemma, we prove that a complete graph $K_{2k}$, where $k \geq 2$, cannot have a sharp nearly equitable $2$-edge-coloring.

\begin{lemma}\label{near-even}
In a complete graph $K_{2k}$, where $k \geq 2$, any nearly equitable $2$-edge-coloring is indeed an equitable $2$-edge-coloring.
\end{lemma}
\begin{proof}
Assume that in the complete graph $K_{2k}$ we have a sharp nearly equitable $2$-edge-coloring with colors red and blue. That is, there exist a vertex $v$ such that  $||{r}(v)|-|{b}(v)|| = 2$, implying that the degree of $v$ is even. However, this contradicts the fact that the degree of $v$ is $2k-1$ in $K_{2k}$.
\end{proof}

 By Corollary \ref{lem-odd3no}, we see that the complete graph $K_{4k+3}$ cannot have an equitable $2$-edge-coloring. On the other hand, in 1982 Hilton and de Werra \cite{HiltonWerra82} proved the following:
\begin{lemma}\textbf{\cite{HiltonWerra82}}\label{nearlyequi}
Any graph $G$ has a nearly equitable edge-coloring with $r$ colors, where $r \geq 2$.
\end{lemma}
\noindent
In the following lemma, we show that a complete graph $K_{2k+1}$, {where $k\ge 1$}, may have a sharp nearly equitable $2$-edge-coloring.

\begin{lemma}\label{rem1}
For each $k \geq 1$, there exists a complete graph $K_{2k+1}$ with a sharp nearly equitable $2$-edge-coloring.
\end{lemma}
\begin{proof}
By  Lemma \ref{nearlyequi} and Corollary \ref{lem-odd3no}, we deduce that complete graphs $K_{4k+3}$ have a nearly equitable $2$-edge-coloring, but do not have an equitable $2$-edge-coloring. Thus, any nearly equitable $2$-edge-coloring is sharp in $K_{4k+3}$.
In a complete graph $K_{4k+1}$, we can obtain a sharp nearly equitable $2$-edge-coloring as follows: let $\chi$ be an equitable $2$-edge-coloring of $K_{4k+1}$ with colors red and blue; that is, each vertex of $K_{4k+1}$ is incident to $2k$ red edges and $2k$ blue edges. We define a new $2$-edge-coloring $\psi$ by interchanging the color of an edge, say from red to blue. It is easy to see that $\psi$ is a sharp nearly equitable $2$-edge-coloring since two vertices of $K_{4k+1}$ are incident to $2k-1$ red edges and $2k+1$ blue edges.
\end{proof}

In the following, we prove that for odd $n \ge 13$, any complete graph $K_n$ having a sharp nearly equitable $2$-edge-coloring has a monochromatic cycle cover with at most two cycles of sizes $\lfloor n/2 \rfloor$ and $\lceil n/2\rceil$ with a single color. Note here that if the complete graph $K_{2k+1}$, for $k \geq 2$, has a nearly equitable $2$-edge-coloring that is also an equitable $2$-edge-coloring, then  by Lemma \ref{lem-odd1} there exists a monochromatic Hamiltonian cycle $C_{2k+1}$.

\begin{lemma}\label{lem-odd3}
For a complete graph $K_{2k+1}$, $k \geq 6$, having a sharp nearly equitable $2$-edge-coloring, there exists a monochromatic cycle cover in a single color of the form $C_{k+1} + C_{k}$ or $C_{2k+1}$. In particular, there exists a cycle cover $C_{k+1} + C_{k}$ in a single color if $K_{2k+1}$ has a subgraph of the form $K_{k+1} + K_k$ induced by a color, and a Hamiltonian cycle $C_{2k+1}$ otherwise.
\end{lemma}

\begin{proof}
Let $\chi$ be a sharp nearly equitable $2$-edge-coloring of the complete graph $K_{2k+1}$ with colors red and blue. Since all vertices of $K_{2k+1}$ have even degree $2k$, $||{r}(v)| -|{b}(v)||$ is either $0$ or $2$. Notice that there must be at least one vertex $v$ with $||{r}(v)| - |{b}(v)|| = 2$ since  $\chi$ is a sharp nearly equitable $2$-edge-coloring. Indeed, by the proof of Lemma \ref{rem1} there exist at least two vertices $u$ and $v$ such that $||{r}(u)| - |{b}(u)|| = ||{r}(v)| - |{b}(v)|| = 2$.

We consider the subgraphs $K_{2k+1}^r$ and $K_{2k+1}^b$ induced by red and blue edges, respectively, in $K_{2k+1}$. For both of them, the minimum degree is at least $k-1$ and  the maximum degree is at most $k+1$; i.e., $\delta \geq k-1$ and $\Delta \leq k+1$. In the case where both $K_{2k+1}^r$ and $K_{2k+1}^b$ are regular graphs, i.e., without loss of generality $(k-1)$-regular and $(k+1)$-regular graphs, respectively, by Theorem \ref{di} we have a Hamiltonian cycle in the $(k+1)$-regular graph $K_{2k+1}^b$. Therefore, we have a monochromatic cycle cover of the form $C_{2k+1}$ in this case.

We then suppose that neither $K_{2k+1}^r$ nor $K_{2k+1}^b$ are regular graphs, i.e., we have $\delta \neq \Delta$ for both subgraphs. Then we have two cases for the minimum degree, namely $\delta = k-1$ or $\delta = k$.

\underline{\textbf{Case $1$:}} Assume that $\delta = k$ for one of $K_{2k+1}^r$ and $K_{2k+1}^b$, say $K_{2k+1}^r$. If $K_{2k+1}^r$ is disconnected, then each component has to have at least $k+1$ vertices, which contradicts the order being $2k+1$. Hence, such a graph has to be connected. Moreover, $K_{2k+1}^r$ has $\Delta = k+1$ because it is not regular. Thus, $K_{2k+1}^r$ is a connected graph with $\delta = k$ and $\Delta = k+1$ on $2k+1$ vertices. By Theorem \ref{buyuk}, $K_{2k+1}^r$ has a Hamiltonian cycle if it is neither the union of two complete graphs $K_{k+1}$ with one common vertex nor the join of an independent set of size $k+1$ with any graph $G_{k}$ with order $k$. First, $K_{2k+1}^r$ cannot be the union of two complete graphs $K_{k+1}$ with one common vertex since $\Delta = k+1$. Therefore, assume that $K_{2k+1}^r$ is the the join of an independent set of size $k+1$ with some graph $G_{k}$. Since $\Delta = k+1$, $G_{k}$ has to be an independent set; therefore, $K_{2k+1}^r$ is isomorphic to the complete bipartite graph $K_{k+1, k}$. Since $K^r_{2k+1}$ has odd order, a Hamiltonian cycle of it is an odd cycle, which contradicts the fact that it is a bipartite graph. Indeed, $K^r_{2k+1}$ cannot have any cycle cover in this case since any cycle cover in $K^r_{2k+1}$ has to have at least one odd cycle. Therefore, $K_{2k+1}^r$ has a Hamiltonian cycle unless it is isomorphic to $K_{k+1, k}$. In the case when $K_{2k+1}^r$ is isomorphic to $K_{k+1, k}$, consider $K_{2k+1}^b$, i.e., the complement of $K_{k+1,k}$.  Observe that $K_{2k+1}^b$ is the disjoint union of two complete graphs $K_{k+1}$ and $K_k$ with $\delta = k-1$ and $\Delta = k$; hence, there exists a monochromatic cycle cover of the form $C_{k+1} + C_{k}$ in $K_{2k+1}^b$. Therefore, we have a monochromatic cycle cover of the form either $C_{2k+1}$ or $C_{k+1} + C_{k}$ in this case. In particular, there exists a  monochromatic cycle cover $C_{k+1} + C_{k}$  if $K_{2k+1}$ has a subgraph $K_{k+1} + K_k$ induced by one of the colors.

\underline{\textbf{Case $2$:}}
Assume that  $\delta = k-1$ for both $K_{2k+1}^r$ and $K_{2k+1}^b$. Then, both of them have $\Delta = k+1$ since they are complements of each other in $K_{2k+1}$.  Assume that such a graph is disconnected. Since $\delta = k-1$, each component has at least $k$ vertices. Since $\Delta=k+1$, at least one component has $k+2$ vertices. Then the order of the graph has to be at least $2k+2$, contradiction. Therefore, a graph with $\delta = k-1$ and $\Delta = k+1$ on $2k+1$ vertices has to be connected. Thus, both $K_{2k+1}^r$ and $K_{2k+1}^b$ are connected graphs with $\delta = k-1$ and $\Delta = k+1$.

\noindent \underline{\textbf{Case $2a$:}}
Assume that at least one of $K_{2k+1}^r$ and $K_{2k+1}^b$ has connectivity $1$; that is, there exists a cut vertex $x$, say in $K_{2k+1}^r$. Then, the subgraph $K_{2k+1}^r - x$ has exactly two components since each component has to have at least $k-1$ vertices and $k\ge 6$. Then $K_{2k+1}^r - x$ has two components $A$ and $B$ such that $k-1 \leq |A|, |B| \leq k+1$ and $|A| + |B| = 2k$. Hence, we have $\delta \geq k-2$ and $\Delta \leq k$ for both $A$ and $B$. We now consider two disjoint and complementary cases:
\begin{itemize}
  \item Assume that $|A| = k-1$ and $|B| = k+1$. Then $A$ is the complete graph $K_{k-1}$ all of whose vertices are adjacent to $x$  in $K_{2k+1}^r$, since $\delta(K_{2k+1}^r) = k-1$. That is, we have a complete subgraph $K_k$ in $K_{2k+1}^r$, and hence a cycle $C_k$ that consists of all vertices of $A$ and $x$. On the other hand, $B$ is then a subgraph with $\delta \geq k-2$ on $k+1$ vertices. Since $k \ge 6$, by Theorem \ref{di}, $B$ has a Hamiltonian cycle $C_{k+1}$. This altogether constitutes a monochromatic cycle cover of the form $C_{k} + C_{k+1}$.

  \item  Assume that $|A| = |B| = k$. Since $\delta \ge k-2$ for both $A$ and $B$ and $k\ge 6$, both $A$ and $B$ have a Hamiltonian cycle $C_k$. Besides, since $|A|=|B|=k$, all vertices in $A$ and $B$ have degree at most $k$ in $K_{2k+1}^r$. Then, the only vertex having maximum degree $k+1$ in $K_{2k+1}^r$ is the cut vertex $x$. It follows that $x$ is adjacent to at least $\lceil (k+1)/2 \rceil$ vertices of either $A$ or $B$, say $A$. Since $x$ is adjacent to more than half of the vertices of $A$, $x$ must be adjacent to two consecutive vertices $y_1$ and $y_2$ of the Hamiltonian cycle $C_{k}$ in $A$. By using the path $y_1 - x - y_2$ instead of the edge $y_1 - y_2$ in $C_k$, we can construct a cycle $C_{k+1}$ covering all vertices of $A$ and $x$ in $K_{2k+1}^r$.
\end{itemize}
Therefore, in this case we have a monochromatic cycle cover of the form $C_{k} + C_{k+1}$.

\noindent\underline{\textbf{Case $2b$:}}
Assume that both subgraphs $K_{2k+1}^r$ and $K_{2k+1}^b$ have connectivity at least $2$. Recall that $\delta = k-1$ and $\Delta = k+1$. Let $\alpha_r$ and $\alpha_b$ be the independence numbers of $K_{2k+1}^r$ and $K_{2k+1}^b$, respectively. By Lemma \ref{independencenumber}, $\alpha_r$ and $\alpha_b$ must be strictly less than $n - \delta = (2k+1) - (k-1) = k+2$  because neither of $K_{2k+1}^r$ and $K_{2k+1}^b$ can be the complete bipartite graph $K_{k+2,k-1}$ since $\Delta=k+1$. We now consider the following disjoint and complementary cases:
\begin{itemize}
  \item Assume that at least one of $\alpha_r$ and $\alpha_b$ is less than or equal to the minimum degree $\delta = k-1$, say $\alpha_r \leq k-1$. Then, since $k\ge 6$, we have $ k-1 = \delta \geq \max\{(n+2)/3,k-1\} = \max\{(2k+1+2)/3, \alpha_r\}$. Therefore, by Theorem \ref{n-w2}, $K_{2k+1}^r$ has a Hamiltonian cycle.
  \item Assume that $\alpha_r = \alpha_b = k$. Then, there exists an independent set of size $k$ and a clique of size $k$ in $K_{2k+1}^r$. Let us partition the vertices of ${K_{2k+1}^r}$ into the sets $V_1$, which is an independent set of size $k$, and $V_2$, which consists of the remaining $k+1$ vertices. All vertices except possibly one vertex of the clique of size $k$ must lie in $V_2$ since $V_1$ is an independent set in $K_{2k+1}^r$. Then, each vertex of this clique lying in $V_2$, i.e., at least $k-1$ vertices, is adjacent to at most three vertices in $V_1$ since its degree is at most $k+1$. Hence, the number of edges joining the vertices of $V_1$ and $V_2$ is at most $3(k-1) + 2k = 5k-3$  by counting edges leaving $V_2$, and at least $k(k-1)=k^2 - k$ by counting edges leaving $V_1$. Since the inequality $k^2 - k > 5k-3$ always holds when $k \geq 6$, i.e., the minimum number of edges leaving $V_1$ is greater than the maximum number of edges leaving $V_2$ when $k \geq 6$, we have a contradiction. Therefore, we cannot have $\alpha_r = \alpha_b = k$.
  \item Assume that one of $\alpha_r$ and $\alpha_b$ is $k+1$ and the other is $k$, say $\alpha_r = k+1$ and $\alpha_b = k$. Then, there exists an independent set of size $k+1$ and a clique of size $k$ in $K_{2k+1}^r$. Let us partition the vertices of ${K_{2k+1}^r}$ into the sets $V_1$, which is an independent set of size $k+1$, and $V_2$, which consists of the remaining $k$ vertices. Since all vertices except possibly one vertex of the clique of size $k$ must lie in $V_2$, each vertex of this clique lying in $V_2$ is adjacent to at most three vertices in $V_1$. Hence, the number of edges joining the vertices of $V_1$ and $V_2$ is at most $3(k-1) + (k+1) = 4k-2$ by counting edges leaving $V_2$, and at least $(k+1)(k-1)=k^2 - 1$ by counting edges leaving $V_1$. Since the inequality $k^2 - 1 > 4k-2$ always holds when $k \geq 6$, we have a contradiction.
  \item Assume that $\alpha_r = \alpha_b = k+1$. In a similar way to the previous cases, there exists an independent set of size $k+1$ and a clique of size $k+1$ in both $K_{2k+1}^r$ and $K_{2k+1}^b$. Let us partition the vertices of ${K_{2k+1}^r}$ into the sets $V_1$, which is an independent set of size $k+1$, and  $V_2$, which consists of the remaining $k$ vertices. Since all vertices except exactly one vertex of the clique of size $k+1$ must lie in $V_2$, each vertex of this clique lying in $V_2$ is adjacent to at most two vertices in $V_1$. Hence, the number of edges joining the vertices of $V_1$ and $V_2$ is  at most $2k$ by counting edges leaving $V_2$, and at least $k(k-1)+k=k^2$ by counting edges leaving $V_1$. Note here that at least one vertex of $V_1$ is adjacent to all vertices in $V_2$. Since the inequality $k^2 > 2k$ always holds when $k \geq 6$, we have a contradiction.
\end{itemize}
Therefore, at least one of $\alpha^r$ and $\alpha^b$ must be less than or equal to the minimum degree $\delta = k-1$. Hence, in this case we have a monochromatic cycle cover of the form $C_{2k+1}$ where $k \geq 6$.

\end{proof}

By combining Lemmata \ref{near-even} and \ref{lem-odd3}, we obtain the following result:

\begin{theorem}\label{combo-near}
 A complete graph $K_{n}$ with $n\ge 13$ and a nearly equitable $2$-edge-coloring has a monochromatic  cycle cover in a single color with at most two cycles, which have sizes ${\lfloor n/2\rfloor}$ and ${\lceil n/2\rceil}$.
\end{theorem}

Hence, we obtain the second main result of this section as follows:

\begin{corollary}\label{main-near}
The solution of the \textsc{MinRC3} problem equals zero for  complete graphs with at least 13 vertices and a nearly equitable $2$-edge-coloring.
\end{corollary}

\section{Algorithm for \textsc{MinRC3}}\label{sec:algo}
\newcommand{\alg}{\textsc{MonochromaticCycleCoverAlgorithm \textit{(MCCA)}}}
\newcommand{\equOdd}{\textsc{EquitableWithOddOrder}}
\newcommand{\equEven}{\textsc{EquitableWithEvenOrder}}
\newcommand{\nearEquOdd}{\textsc{NearEquitableWithOddOrder}}

\newcommand{\Dirac}{\textsc{DiracHamiltonian}}
\newcommand{\Closure}{\textsc{ClosureHamiltonian}}
\newcommand{\ExtensionDirac}{\textsc{ExtensionDiracHamiltonian}}

\alglanguage{pseudocode}

In this section we present an algorithm referred to as the \textsc{MonochromaticCycleCoverAlgorithm \textit{(MCCA)}}, which, given a complete graph $K_n$ with a $2$-edge-coloring, returns either a monochromatic cycle cover $\textbf{C}$ or ``NONE". Although \textsc{\textit{MCCA}} may in general return ``NONE" for a complete graph with a $2$-edge-coloring $\chi$, we will show that except possibly on a complete graph with four vertices, it returns a monochromatic cycle cover if  $\chi$ is an equitable $2$-edge-coloring. Furthermore, except for some special cases, \textsc{\textit{MCCA}} mostly (but not always) returns a monochromatic cycle cover if $\chi$ is a nearly equitable $2$-edge-coloring.

We first consider a complete graph $K_{n}$ of even order $n$, say $n=2k$. By Lemma \ref{near-even}, any nearly equitable $2$-edge-coloring is indeed an equitable $2$-edge-coloring in $K_{2k}$. Hence, the algorithm \textsc{\textit{MCCA}} works identically for both equitable and nearly equitable $2$-edge-colorings in $K_{2k}$. We then consider the case where $K_{2k}$ has an equitable $2$-edge-coloring. Given a subgraph $G$ of $K_{2k}$ induced by a color, the algorithm tests $G$ for $\delta \geq n/2$, which is Dirac's sufficiency condition for hamiltonicity given in Theorem \ref{di}. Once ${G}$ passes the test, the algorithm constructs a Hamiltonian cycle via the function $\Dirac$, which builds a Hamiltonian cycle by following the proof of Theorem \ref{di}. According to the proof of Lemma \ref{lem-even}, a monochromatic cycle cover, in particular a Hamiltonian cycle, is obtained when the subgraph $G$ induced by a color is a disconnected or a regular graph in this case. If $G$ fails to satisfy the condition $\delta \geq n/2$, then  the algorithm builds the closure $G^*$ of $G$ and tests $G^*$ for being a complete graph according to Bondy-Chv\'{a}tal's hamiltonicity condition given in Theorem \ref{bo}. Once $G^*$ passes the test, the algorithm constructs a Hamiltonian cycle via the function $\Closure$, which builds a Hamiltonian cycle by following the proof of Theorem \ref{bo}. Indeed, in the rest of the proof of Lemma \ref{lem-even} we use Theorems \ref{bo}, \ref{ch} and \ref{m-m}, which give sufficiency conditions for closure and hamiltonicity of $G$. Hence, the function $\Closure$ will be sufficient to construct a monochromatic cycle cover, in particular a Hamiltonian cycle, in order to complete the rest of this case.

We now consider a complete graph $K_{n}$ of odd order $n$, say $n=2k+1$. In this case, we first assume that $K_{2k+1}$ has an equitable $2$-edge-coloring. Indeed, since the complete graph $K_{4k+3}$ does not have an  equitable $2$-edge-coloring by Corollary \ref{lem-odd3no}, we only need to consider a complete graph $K_{4k+3}$ with an equitable $2$-edge-coloring. Given a subgraph $G$ of $K_{4k+3}$ induced by a color, the algorithm tests $G$ for  $\delta \geq \lfloor n/2 \rfloor$, which is Büyükçolak's sufficiency condition for hamiltonicity given in Theorem \ref{buyuk}. Once ${G}$ passes the test, the algorithm constructs a Hamiltonian cycle via the function $\ExtensionDirac$, which builds a Hamiltonian cycle by following the proof of Theorem \ref{buyuk} given in \cite{Buyukcolak}. According to the proof of Lemma \ref{lem-odd1}, a monochromatic cycle cover, particularly a Hamiltonian cycle, is obtained.

Let us consider the case where a complete graph $K_{2k+1}$ has a nearly equitable $2$-edge-coloring. By Lemma \ref{rem1} and \ref{lem-odd3}, a complete graph $K_{4k+1}$ may have equitable and nearly equitable $2$-edge-colorings in different forms, whereas a complete graph $K_{4k+3}$ can only have a nearly equitable $2$-edge-coloring. Therefore, if $K_{4k+1}$ has a nearly equitable $2$-edge-coloring $\chi$, which is also an equitable $2$-edge-coloring, then the algorithm constructs a monochromatic cycle cover, in particular a Hamiltonian cycle, by considering $\chi$ as an equitable $2$-edge-coloring. We then consider the case where a complete graph $K_{2k+1}$ has a sharp nearly equitable $2$-edge-coloring. Given a subgraph $G$ of $K_{2k+1}$ induced by a color, the algorithm works in the following way:
\begin{itemize}
  \item The algorithm tests $G$ for $\delta \geq n/2$ and constructs a Hamiltonian cycle via the function $\Dirac$ if $G$ passes the test. According to the proof of Lemma \ref{lem-odd3}, in this case a monochromatic cycle cover, in particular a Hamiltonian cycle, is obtained when the subgraph $G$ induced by a color is a regular graph.
  \item Otherwise, the algorithm then tests $G$ for $\delta \geq \lfloor n/2 \rfloor$, which is Büyükçolak's sufficiency condition for hamiltonicity given in Theorem \ref{buyuk}. If ${G}$ passes the test, then the algorithm constructs either a cycle cover $C_{k+1}+ C_k$ by following the proof of Theorem \ref{buyuk} given in \cite{Buyukcolak} or a Hamiltonian cycle via the function $\ExtensionDirac$. According to the proof of Lemma \ref{lem-odd3}, a monochromatic cycle cover of the form $C_{k+1}+ C_k$ is obtained in the complement of $G$ when the subgraph $G$ is a complete bipartite graph $K_{k+1,k}$ in this case. Indeed, if $G = K_{k+1,k}$, the algorithm constructs two cycles $C_1$ and $C_2$ with the vertices of degree  $\lceil n/2\rceil$ and the vertices of degree $\lfloor n/2\rfloor$, respectively. Notice that the order of vertices in $C_1$ and $C_2$ makes no difference for the Hamiltonian cycle since these sets of vertices form two distinct complete graphs in the complement of $G$.
  \item Otherwise, the algorithm tests $G$ for a cut vertex $x$. If $G$ passes the test, then the algorithm constructs a cycle cover of the form $C_{k+1}+ C_k$ in two different ways by using the structure of $G-x$ given in the proof of Lemma \ref{lem-odd3}.
\end{itemize}
Notice that the algorithm returns ``NONE" for a complete graph with a nearly equitable $2$-edge-coloring $\chi$ if both subgraphs induced by the colors are $2$-connected with $\delta = k-1$ and $\Delta =k+1$ (Case 2b in the proof of Lemma \ref{lem-odd3}). {In other words, the algorithm remains inconclusive in this case since the corresponding part of the proof is not constructive.}

\begin{algorithm}[H]
\caption{\alg}\label{alg:MCCA}
\begin{algorithmic}[1]
\Require {A complete graph $K_n$ of order $n$ with a $2$-edge-coloring $\chi$}
\Ensure { $\textbf{C}$ is a monochromatic cycle cover of $K_n$ in a single color}

\State $G_1 \gets$ the subgraph of $K_n$ induced by red.
\State $G_2 \gets$ the subgraph of $K_n$ induced by blue.

\State $\delta_1 \gets$ the minimum degree of $G_1$.
\State $\delta_2 \gets$ the minimum degree of $G_2$.

\State $\Delta_1 \gets$ the maximum degree of $G_1$.
\State $\Delta_2 \gets$ the maximum degree of $G_2$.

\For{$i=1$ to $2$}

\If{$\delta_i \geq n/2$}
\State  $C \gets $ \Call{DiracHamiltonian}{$G_i,\delta_i$}
\If {$C \neq$ NONE} \Return $C$.  \EndIf
\EndIf

\State $G_i^* \gets $ closure of $G_i$

\If{$G_i^*$ is a complete graph}
\State $C \gets $ \Call{ClosureHamiltonian}{$G_i$,$G_i^*$}
\If {$C \neq$ NONE} \Return $C$.  \EndIf
\EndIf

\If{$\delta_i = \lfloor n/2\rfloor$}

\If{$G_i$ is a complete bipartite graph $K_{\lceil n/2\rceil,\lfloor n/2\rfloor}$}
\State  $C_1 \gets $ the vertices of degree $\lceil n/2\rceil$
\State  $C_2 \gets $ the vertices of degree $\lfloor n/2\rfloor$
\State  $C \gets   C_1 + C_2$
\If {$|V(C)| = |V(G_i)|$} \Return $C$.  \EndIf
\Else
\State  $C \gets $ \Call{ExtensionDiracHamiltonian}{$G_i,\delta_i$}
\EndIf
\If {$C \neq$ NONE} \Return $C$.  \EndIf
\EndIf

\If{$G_i$ has a cut vertex $x$ and if $k\ge 4$}
\State Let $A$ and $B$ be two components of $G_i - x$ such that $|B|\geq|A|$

\If{$|B|>|A|$}
\State  $C_1 \gets $ the vertices of $A$ and $x$
\State  $C_2 \gets $ \Call{DiracHamiltonian}{$B,\delta_i-1$}
\State  $C \gets   C_1 + C_2$
\If {$|V(C)| = |V(G_i)|$} \Return $C$.  \EndIf
\ElsIf{$|B|=|A|$}
\State  $C_1 \gets $ \Call{DiracHamiltonian}{$A,\delta_i-1$}
\State  $C_2 \gets $ \Call{DiracHamiltonian}{$B,\delta_i-1$}
\State Let $x$ be adjacent to more vertices of $A$ than $B$.
\State Let $P=x_0x_1 ... x_{k-1} \gets C_1$.
\For{$j = 0$ to $k-1$}
\If{$xx_{j} \in E(G_i)$ and  $xx_{j+1} \in E(G_i)$}
\State \Return $C_1=(x_0 ... x_{j}xx_{j+1} ... x_{k-1})$.
\EndIf
\EndFor
\State  $C \gets   C_1 + C_2$
\If {$|V(C)| = |V(G_i)|$} \Return $C$.  \EndIf
\EndIf
\EndIf
\EndFor\\
\Return NONE
\end{algorithmic}
\end{algorithm}

By using the constructive nature of Dirac's original proof for Theorem \ref{di}, the work in \cite{Web1} presents a polynomial-time algorithm for finding Hamiltonian cycles in graphs that satisfy the condition of Theorem \ref{di}, i.e.,  having at least three vertices and minimum degree at least half the total number of vertices. For the sake of completeness, in Algorithm \ref{alg:Dirac} we give a function which produces a Hamiltonian cycle under the condition of Theorem \ref{di}.

In Algorithm \ref{alg:Dirac}, the function $\Dirac$ first builds a maximal path by starting with an edge and then extending it in both directions as long as this is possible. Afterwards, the function closes the path to a cycle and then tries to find a larger path by adding to the cycle a new vertex and opening it back to a path. By the minimum degree condition $\delta \geq n/2$, any maximal path can be closed to a cycle and it is possible to extend a closed cycle to a larger path. Finally, the function builds a Hamiltonian path and then a Hamiltonian cycle.

\begin{algorithm}[H]
\caption{DiracHamiltonian}\label{alg:Dirac}
\begin{algorithmic}[1]
\Function{DiracHamiltonian}{$G$, $\delta$}
\Require{$\delta \geq |V(G)|/2$}
\Ensure{return a Hamiltonian cycle $C$ or ``NONE"}

\State $P \gets$ a trivial path in $G$.
\Repeat
\While{$P$ is not maximal}
\State Append an edge to $P$.
\Comment{$P$ is a maximal path in $G$.}
\EndWhile
\State Let $P=x_0x_1 ... x_k$.
\For{$i = 0$ to $k-1$}
\If{$x_0x_{i+1} \in E(G)$ and  $x_ix_{k} \in E(G)$}
\State \Return $C=(x_0,x_{i+1} ... x_{k-1}, x_k,x_i,x_{i-1} ... x_1,x_0)$. \\
\Comment{The existence of such an index $i$ is guaranteed by minimum degree condition.}
\EndIf
\EndFor
\If{$C \neq$ NONE and $|V(C)| \neq |V(G)|$}
\State Let $e$ be an edge with exactly one endpoint in $C$.
\State Let $e'$ be an edge of $C$ incident to $e$ \Comment{There are two such edges.}
\State $P \gets C + e - e'$
\EndIf
\Until{$|V(C)|=|V(G)|$ or $C=$ NONE}
\State \Return $C$.
\EndFunction
\end{algorithmic}
\end{algorithm}

As a result of the constructive nature of Bondy-Chv\'{a}tal's proof for Theorem \ref{bo}, there exists a polynomial-time algorithm which produces a Hamiltonian cycle in graphs whose closure is a complete graph. For the sake of completeness, we give such an algorithm  in Algorithm \ref{alg:Closure}.

In Algorithm \ref{alg:Closure}, the function first arbitrarily arranges all vertices in a cycle since the closure is complete. Note that this cycle is Hamiltonian since it contains all vertices of the graph. If all edges of the cycle are already in the graph, then we are done. Otherwise, there exists an edge $e$ which is in the closure but not in the graph. The function opens this Hamiltonian cycle to a Hamiltonian path $P=x_0x_1 ... x_k$ by removing $e$, and builds a new Hamiltonian cycle from this Hamiltonian path using edges $x_0x_{i+1}$ and $x_kx_{i}$ and removing edge $x_ix_{i+1}$ in the graph for some $2 \leq i \leq k-1$. The existence of such an edge is guaranteed by definition of closure, i.e., $d(x_0) + d(x_k) \geq n$. After repeating this process for each edge which is in the closure but not in the graph, the function constructs a Hamiltonian cycle in the graph.

\begin{algorithm}[H]
\caption{ClosureHamiltonian}\label{alg:Closure}
\begin{algorithmic}[1]
\Function{ClosureHamiltonian}{$G$,$cl(G)$}
\Require{$cl(G)$ is a complete graph}
\Ensure{return a Hamiltonian cycle $C$ or ``NONE"}

\State $C \gets$ a Hamiltonian cycle in $cl(G)$. \\
\Comment{$C$ can be obtained by arbitrarily arranging all vertices.}

\For{$j = 1$ to $k$}

\If{$e_i$ is not an edge in $G$}

\State $P \gets C - e_j$
\Comment{$P$ is a Hamiltonian path in $\bar{G}$.}
\State Let $P=x_0x_1 ... x_k$.
\For{$i = 2$ to $k-1$}
\If{$x_0x_{i+1} \in E({G})$ and  $x_ix_{k} \in E({G})$}
\State $C \gets (x_0,x_{i+1} ... x_{k-1}, x_k,x_i,x_{i-1} ... x_1,x_0)$.
\EndIf
\EndFor
\EndIf

\EndFor

\Repeat
\State Let $E=\{e_1, e_2, ...  ,e_k\}$ be the edge set of $C$.

\Until{$E \subseteq E(G)$ or $C=$ NONE}
\State \Return $C$.
\EndFunction
\end{algorithmic}
\end{algorithm}

Since the function $\ExtensionDirac$ and its constructive structure are explicitly stated in \cite{Buyukcolak}, we do not give the algorithm $\ExtensionDirac$ here. We refer to \cite{Buyukcolak} for details.

\section{Conclusion}

In this work, we show that there exists a monochromatic cycle cover in complete graphs with at least $13$ vertices and a nearly equitable $2$-edge-coloring. Hence, we conclude that the minimum reload cost is  zero in these graphs. In general, all proofs except one part in this paper are constructive. Then, we provide a polynomial-time algorithm that constructs a monochromatic cycle cover, in particular a Hamiltonian cycle or two cycles whose sizes differ by at most one, in complete graphs with a nearly equitable $2$-edge-coloring. This algorithm builds a monochromatic cycle cover in all complete graphs with an equitable $2$-edge-coloring, whereas it may remain inconclusive in some complete graphs with a sharp nearly equitable $2$-edge-coloring. In particular, the algorithm remains inconclusive for the case where both subgraphs induced by a color in a complete graph of odd order $2k+1$ with a sharp nearly equitable $2$-edge-coloring is $2$-connected with $\delta = k-1$ and $\Delta = k+1$.

We believe that the \textsc{MinRC3} problem may have solution zero for other types of $2$-edge colorings in complete graphs because of the insight provided by this work. As a future work, we  plan to study the \textsc{MinRC3} problem in complete graphs with a $2$-edge coloring in general and design an algorithm that not only determines whether a monochromatic cycle cover exists but also constructs a monochromatic cycle cover whenever it exists.

\bibliography{sample}

\end{document}